\documentclass[a4paper]{article}
\usepackage[latin1]{inputenc} 
\usepackage[T1]{fontenc} 
\usepackage{RRA4,RRthemes}
\usepackage{hyperref}
\graphicspath{{./logos/}}
\RRdate{October 2011}

\RRauthor{
Ichrak Amdouni 
  \and
Cedric Adjih 
 \and
Pascale Minet
}
\authorhead{Amdouni \& Adjih \& Minet}
\RRtitle{Coloriage des r\'eseaux de capteurs en grille: La m\'ethode de coloriage bas\'ee sur les vecteurs}
\RRetitle{On the Coloring of Grid Wireless
Sensor Networks: the Vector-Based Coloring Method}
\titlehead{Node coloring based on VCM}
\RRresume{Le coloriage des graphes est utilis\'e dans les r\'eseaux sans fil afin d'optimiser les ressources du r\'eseau:  la bande passante et l'\'energie. Les noeuds du r\'eseau acc\`edent au m\'edium en fonction de leur couleur. L'algorithme de coloriage a la charge d'assurer que les noeuds interef\'erents n'aient pas la m\^eme couleur. Dans ce rapport de recherche, nous nous concentrons sur les r\'eseaux de capteurs sans fil ayant une topologie en grille. Comment un algorithme de coloriage peut profiter de la r\'egularit\'e de cette topologie pour r\'ealiser un coloriage p\'eriodique optimal, c'est \`a dire un coloriage avec le nombre minimal de couleurs?
 Nous proposons la {\em M\'ethode des Vecteurs}, not\'ee {\em VCM}, une nouvelle m\'ethode qui est capable de fournir un coloriage \`a $h$ sauts, p\'eriodique et optimal quelque soit $h \ge 1$ et quelque soit la port\'ee radio. Cette m\'ethode consiste \`a d\'eterminer quels noeuds de la grille peuvent utiliser la m\^eme couleur sans cr\'eer des interf\'erences entre ces noeuds tout en minimisant le nombre de couleurs utilis\'ees. Nous comparons le nombre de couleurs utilis\'ees par la {\em M\'ethode des Vecteurs}, \`a celui obtenu par un algorithme distribu\'e qui affecte les couleurs selon la priorit\'e donn\'ee par la ligne ou la colonne.
Nous fournissons aussi des bornes sur le nombre de couleurs des coloriages g\'en\'eraux optimaux de la grille, et prouvons que les coloriages
p\'eriodiques (et donc VCM) sont asymptotiquement optimaux.
 Enfin, nous discutons l'applicabilit\'e de cette m\'ethode \`a un r\'eseau sans fil r\'eel.
}
\RRabstract{Graph coloring is used in wireless networks to optimize network resources: bandwidth and energy. Nodes access the medium according to their color. It is the responsibility of the coloring algorithm to ensure that interfering nodes do not have the same color. In this research report, we focus on wireless sensor networks with grid topologies. How does a coloring algorithm take advantage of the regularity of grid topology to provide an optimal periodic coloring, that is a coloring with the minimum number of colors? We propose the {\em Vector-Based Coloring Method}, denoted {\em VCM}, a new method that is able to provide an optimal periodic coloring for any radio transmission range and for any $h$-hop coloring, $h\geq1$. This method consists in determining at which grid nodes a color can be reproduced without creating interferences between these nodes while minimizing the number of colors used. We compare the number of colors provided by {\em VCM} with the number of colors obtained by a distributed coloring algorithm with line and column priority assignments. We also provide bounds on the number of colors of optimal general colorings of the infinite grid, and show that periodic colorings (and thus VCM) are asymptotically optimal.
Finally, we discuss the applicability of this method to a real wireless network.
}
\RRmotcle{Coloriage des graphes, r\'eseaux de capteurs sans fil, grilles, coloriage p\'eriodique, M\'ethode des Vecteurs, coloriage valide, coloriage optimal, bornes sur le nombre de couleurs.}
\RRkeyword{Graph coloring, wireless sensor networks, grid, periodic coloring, Vector-based Method, valid coloring, optimal coloring, number of colors bounds.}
\RRprojet{Hipercom}  
\RRdomaine{3} 
\RRthemeProj{hipercom} 
\RRdomaineProjBis{hipercom} 
\RCParis 

\usepackage{subfigure}
\usepackage{mathenv}
\usepackage{amssymb}

\def\QED{\mbox{\rule[0pt]{1.5ex}{1.5ex}}} 
\def\proof{\hspace*{-1.4em}{{\itshape Proof: }}} 
\def\endproof{\hspace*{\fill}~\QED\par\endtrivlist\unskip} 
\usepackage{color}
\usepackage{comment}
\usepackage{graphicx}
\usepackage{mathenv}
\usepackage{amssymb}
\usepackage{mathtools}
\newtheorem{definition}{Definition} 
\newtheorem{lemma}{Lemma} 
\newtheorem{property}{Property} 
\newtheorem{theorem}{Theorem} 
\newtheorem{corollary}{Corollary} 
\newtheorem{remark}{Remark} 
\newenvironment{method}[1][Method]{\begin{trivlist}
\item[\hskip \labelsep {\bfseries Method #1:}]}{\end{trivlist}}

\newcommand{\ZZ}{\mathbb{Z}}
\newcommand{\RR}{\mathbb{R}}
\newcommand{\Neigh}[1]{\mathcal{N}(#1)}
\newcommand{\refsec}[1]{Section~\ref{#1}}
\newcommand{\myvec}[1]{u_{#1}}
\newcommand{\XXX}[1]{{\color{red}{XXX:#1}}}

\includecomment{LongVer}
\excludecomment{ShortVer}

\newcommand{\minipagesize}{\textwidth}
\newcommand{\mybox}[1]{
\vspace{1mm}\fbox{\begin{minipage}{\minipagesize}#1\end{minipage}}\vspace{1mm}}

\begin{document}
\RRNo{7756}
\makeRR   
\tableofcontents
\newpage

\section{Context and motivation}\label{contextSec}
Node coloring consists in assigning colors to nodes using a minimum number of colors such that two interfering nodes do not have the same color. It has been proved that the 1-hop coloring problem in \cite{Gar79} and more generally the $h$-hop node coloring problem with $h$ an integer $\geq 1$ is NP-complete~\cite{C83,RR11}. That is why, approximation algorithms are adopted. Among them, FirstFit \cite{Car07} assigns to the uncolored node with the highest priority the first available color. Unfortunately, the number of colors obtained depends on the node coloring order given by the priority. A whole class of results expresses properties related to the \emph{worst-case} performance of approximation algorithms: 
they typically prove that, for any input graph of a given family, the coloring obtained by a given 
algorithm  uses at most $\alpha$ times the optimal number of colors. 
Such an algorithm is denoted an $\alpha$-approximation algorithm. 

Many coloring algorithms have been designed to be used in wireless networks to make the medium access more efficient. Among them, we can cite TRAMA~\cite{Rajendran03}, FLAMA~\cite{Rajendran05}, ZMAC~\cite{ZMAC} which is based on the DRAND~\cite{DRAND} coloring algorithm, TDMA-ASAP~\cite{TDMA-ASAP}, FlexiTP~\cite{FlexiTP} and SERENA~\cite{iwcmc08}. In short, time slots are assigned per node color: two nodes with the same color can transmit in the same time slot without interfering.

In this research report, we focus on wireless sensor networks with grid topologies. These topologies are regular. 
Wireless sensors organized according to a grid topology are 
a natural choice and exist in many real applications.
Often, the grid organization is easy to deploy,
and is efficient in terms of coverage, connectivity and management.
For instance, a grid topology is one of the best methods to ensure sensor 
coverage for surveillance \cite{BFKKP10,CIQC02}. 
It also helps to collect measurements with a uniform spatial sampling such as for
instance in precision agriculture and irrigation as in \cite{MGHMPT08}: when physical phenomena are numerically modeled, measurements from a grid pattern may be a direct input or directly compared to the output of equations solved with the finite element method on a grid --- without requiring additional sensor localization and numerical measurement interpolation.\\
Given such grid topologies, the question (and our problem statement), is:
what is the method to color grids while maximizing the color reuse?
Naturally, a first step towards answering the question, is to apply existing
algorithms to a sample of grids and analyze the output for some sample topologies. 
\begin{LongVer}
\begin{table}[!htb]
\caption{Number of colors for various transmission ranges, grid sizes and priority assignments., for 3-hop coloring}
\label{TableVariousPrioNodeNb}
\begin{center}
\begin{tabular}{|c|c|c|c|c|}
\hline
Grid size & priority assignment & colors &  colors \\
 &  & for range = 1 & for range = 2 \\
\hline \hline
10x10 & line or column& 8* & 30 \\
          \cline{2-4}
          & diagonal &8* & 28 \\
          \cline{2-4}
          & distance to origin &8* & 30 \\
          \cline{2-4}
          & random &13 & 36\\
    \cline{1-4}
20x20 & line or column & 15 & 33\\
           \cline{2-4}
          & diagonal &8* & 29\\
          \cline{2-4}
          & distance to origin &8* & 30\\
         \cline{2-4}
          & random &14 & 41\\
\hline 
\end{tabular}
\end{center}
\end{table}
\end{LongVer}
\newline
\begin{LongVer}
\newpage
We have performed simulations of a 3-hop coloring, because 3-hop coloring is needed when broadcast and immediate acknowledgement of unicast transmission are required.  
\end{LongVer}
The algorithm colors nodes according to their priority order among their neighbors up to $3$-hop. Node priority is given by one of the following heuristics: the position of the node in the line, column, diagonal, its distance to the grid center or random.
Results are summarized in table~\ref{TableVariousPrioNodeNb}. The '*' symbol highlights the optimality of the number of colors used. The reader can refer to \cite{RR11} for the proof of the optimal number of colors needed to color some grids with various transmission ranges. Different radio ranges are simulated.
Mostly, we observe that no priority assignment tested provides the optimal number of colors for any grid size and any radio range.
\newline
This prompts the following questions: Does a priority assignment exist in grid topologies such that the number of colors does not depend on node number but only on radio range?
Is it possible to take advantage of the regularity of grid to design a $h$-hop coloring algorithm able to find the optimal number of colors for any grid with any transmission range value?
%
 Moreover, does a periodic color pattern that can tile the whole topology for a given radio range, exist? What is the optimal number of colors? Is the optimal number of color reached by periodic colorings?

In this research report, we answer these questions by proposing a coloring method, 
the {\em Vector-Based Coloring Method}, denoted {\em VCM} which provides an optimal and valid coloring of grids.
\newline

\begin{LongVer}

The research report is organized as follows.
In Section~\ref{sec:notations-statement}, we introduce some definitions and notations. We present the problem statement and give an intuitive idea of VCM and introduce its main components.
In Section~\ref{periodicCol}, we focus on periodic coloring. In Section~\ref{sec:vcm-color-computation}, we consider periodic colorings where any color is periodically reproduced, and show how VCM computes the color of an node in the grid. In Section~\ref{sec:vcm-validity}, we restrict or study of periodic coloring to valid ones. We propose two methods to check the validity of periodic coloring. Section~\ref{sec:vcm-ovs} allows us to characterize an optimal coloring and give upper and lower bounds on the number of colors, and shows that periodic colorings (and therefore VCM) are asymptotically optimal when the radio range increases\footnote{showing that specifically for grids, there are better techniques than already known general $\alpha$-approximations (even for unit disk graphs, such as \cite{gsw1998}), since here $\alpha = 1 + O(1/R)$ when $R \rightarrow \infty$}.
We then show how to reduce the research space of the generator vectors produced by VCM in order to reduce the complexity of VCM that we evaluate (which is polynomial time). We conclude this section by giving the method to find the optimal valid generator vectors.
In Section~\ref{sec:vcm-practice}, we summarize by giving rules to apply VCM in practice. Results obtained by VCM are given in Section~\ref{resultsSec} and compared with those obtained by different heuristics.  Finally, in Section~\ref{sec:Realwireless}, we discuss how to apply VCM in real wireless sensor networks. We conclude in Section~\ref{conclusion}. In the Annex, the reader will find some mathematical results and grid properties useful for grid coloring.

\end{LongVer}
\begin{ShortVer}
\XXX{UPDATE-THIS: The paper is organized as follows:
in Section~\ref{sec:notations-statement}, we introduce some definitions, notations and the problem statement; in Section~\ref{}
in Section~\ref{optimalColoringSec}, we present the principles and properties of the $Vector Method$; in Section~\ref{}}
\end{ShortVer}

\section{Problem Statement and Overview}\label{sec:notations-statement}
\subsection{Notations, Definitions and Assumptions}
\label{sec:notations}
In the remaining of this research report, we use the following notations:
\begin{itemize}
\item $\vec{UV}$ denotes a vector of extremities the nodes $U$ and $V$.
\newcommand{\mynorm}[1]{|#1|}
\item $\mynorm{\vec{UV}}$ denotes the norm, or length of the vector $\vec{UV}$,
and $d(U,V)$ the euclidian distance between nodes $U$ and $V$.
\item $det(\vec{UV_1},\vec{UV_2})$ is the determinant of the two vectors.
\item $\vec{UV_1} \cdot \vec{UV_2}$ denotes the scalar product of the vectors.
\item $\Lambda(\vec{UV_1},\vec{UV_2})$ denotes the lattice having as a base the vectors $(\vec{UV_1},\vec{UV_2})$. We call $\vec{UV_1}$ and $\vec{UV_2}$ the basis  (or generator) vectors. $\mathcal{P}(\vec{UV_1},\vec{UV_2})$ is the fundamental parallelotope associated to $\Lambda$. In 2 dimensions, $\mathcal{P}$ is a parallelogram. Moreover, the number of nodes in $\mathcal{P}(\vec{UV_1},\vec{UV_2})$ called {\em lattice determinant} is given by $det(\vec{UV_1},\vec{UV_2})$. These notations and definitions are adopted from~\cite{latticeDef}. We also denote $\mathcal{P}_{U}(\vec{UV_1},\vec{UV_2})$ the fundamental
parallelotope translated at node $U$.
\item For $x, y$ in $\ZZ^2$, with $y\neq 0$, we denote "x modulo y" the integer z in $\{0, 1, \ldots |y|-1\}$ such that we have $x \equiv z (mod |y|).$
\end{itemize}
Moreover, the following definitions and assumptions are used:%
\begin{itemize}
\item {\bf Nodes:} The nodes are disposed in a grid. 
Without loss of generality, we assume that
the grid step is $1$, hence the set of nodes is identified by $\ZZ^2$.
\item
{\bf Neighborhood:}
$R$ denotes the radio range and is a real number $\ge 1$. 
The underlying assumption is a unit disk model, assuming that nodes $U$ and $V$
can communicate if their euclidian distance is lower than $R$ and are able to communicate directly in both directions. Hence, the set of neighbors of a node $U$ is:\\
$\Neigh{U} = \{ V \in \ZZ^2  ~|~ 0 < d(U,V) \le R  \}$.
\item {\bf Coloring:} A coloring $\phi$ is a mapping from the nodes $\ZZ^2$ to a set of colors, identified with the set of integers $\{0, \ldots k-1 \}$, with $k$ a positive integer. 
\item {\bf Valid Coloring:}
a coloring is said to be a {\em valid $h$-hop coloring}, with $h$ an integer $\ge1$, 
when all nodes that are less or equal to $h$-hop away are assigned different colors.
\item {\bf Periodic Coloring:} A coloring is denoted ``periodic'' if the set of nodes with the same given color, is the lattice $\Lambda(v_1,v_2)$ for some fixed
vectors $v_1$, $v_2$ of $\ZZ^2$, after it is translated by some amount. In this paper, we consider a ``strict'' definition of periodicity because colors 
might otherwise be periodically repeated with a pattern generating more than a single lattice for instance.
\end{itemize}

\subsection{Problem Statement}
Our general goal is the following: \emph{Find a valid $h$-hop coloring of 
the nodes of the infinite grid $\ZZ^2$, with a minimal number of colors for $h\ge1$}. 
\newline
Because the set of colorings of $\ZZ^2$ is infinite, in this paper we 
restrict ourselves to colorings exhibiting a periodic color pattern
(following the definition of periodic coloring of section~\ref{sec:notations}).
\begin{quote}
{\bf Problem Statement:}
Find one of the valid periodic $h$-hop colorings with the minimum number of colors for $h\ge1$.
\end{quote}

\subsection{Overview}
In this paper, we propose the {\em periodic Vector-Based Coloring Method}, denoted {\bf VCM} that answers the problem
statement in polynomial time. This method consists in producing a periodic coloring obtained by tiling a color pattern. This color pattern is generated by two vectors verifying some conditions that will be detailed in the paper. \\
\newline
\begin{figure}[!h]
\begin{center}
{\includegraphics[width=0.9\linewidth]{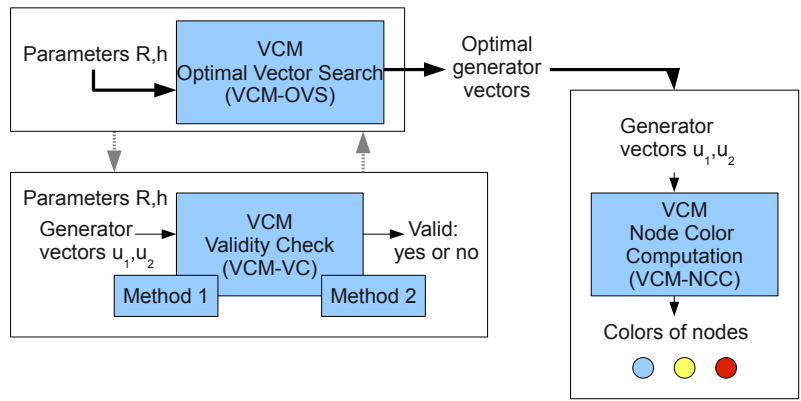}}
\caption{Components of VCM}\label{VCMcomponents}
\end{center}
\end{figure}
As illustrated in Figure~\ref{VCMcomponents}, {\em VCM} is composed of the following components: \\
{\bf C1.} {\bf VCM-NCC} A periodic coloring based on two generator vectors. This component determines how any sensor node selects its color.\\
{\bf C2.} {\bf VCM-VC} An algorithm to check the validity of the periodic coloring: the VCM incorporates two methods to verify the validity of the coloring for a given couple of generator vectors.\\
{\bf C3.} {\bf VCM-OVS} An algorithm to search for optimal generator vectors, that is, yielding the minimum number of colors.
Indeed, we limit the set of candidate vectors to find the vectors providing the optimal number of colors.\\
\begin{LongVer}
\subsection{Intuitive Idea}
The intuitive idea of the method is as follows. As the grid topology presents a regularity in terms of node position, {\em VCM} produces a similar regularity in terms of colors and generates a color pattern that can be periodically reproduced to color the whole grid. 
Our aim is to find such a color pattern that minimizes the number of colors used. More precisely, given a colored node $U$, we determine where its color can be reproduced. 
The method starts by the choice of two
vectors $\vec{UV_1}$ and $\vec{UV_2}$ such that $V_1$ and $V_2$ use the same color as $U$. Of course, the vectors $\vec{UV_1},\vec{UV_2}$ must provide a valid $h$-hop coloring.
The color pattern is given by the set of colors in the finite parallelogram $\mathcal{P}(\vec{UV_1},\vec{UV_2})$ translated at $U$. Hence, the color of $U$ is repeated also at any node $W$ where $\vec{UW}$ is a linear combination of $\vec{UV_1}$ and $\vec{UV_2}$.\\
\end{LongVer}
In the Sections from~\ref{sec:vcm-color-computation} to~\ref{sec:vcm-ovs}, we detail the components of {\em VCM}.\\

\section{Periodic Coloring}\label{periodicCol}\label{principlePeriodic}
Figure~\ref{vectorMethodFig} presents an example of a periodic 3-hop coloring of a $10 \times 10$ grid with a transmission range $R=1$.
\begin{figure}[!h]
\begin{center}
{\includegraphics[width=0.5\linewidth]{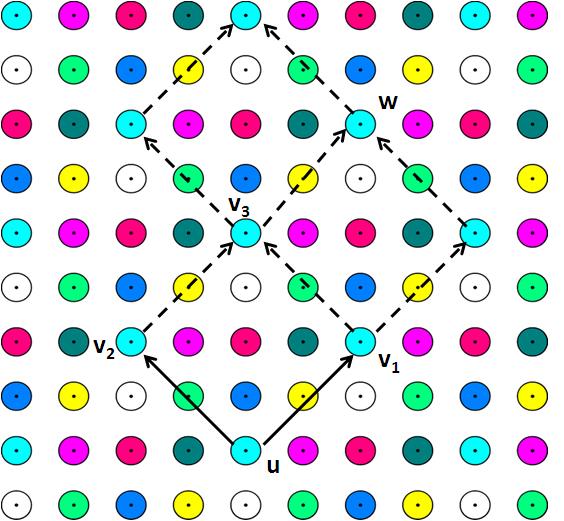}}
\caption{Example of periodic 3-hop coloring (R=1)}\label{vectorMethodFig}
\end{center}
\end{figure}
\newline
The principles of the periodic coloring are:\\
{\bf P1. (Generator vectors)} The two vectors $\vec{UV_1}$ and $\vec{UV_2}$, if linearly independent, generate the parallelogram $\mathcal{P}(\vec{UV_1},\vec{UV_2})$ of the color pattern.\\ 
{\bf P2. (Parallelotope color unicity)} Inside $\mathcal{P}(\vec{UV_1},\vec{UV_2})$, there is no color reuse.\\
{\bf P3. (Lattice color repetition)} Because the periodic coloring is obtained by repeating the color pattern, the nodes 
$W$ such that $\vec{UW}$ is a linear combination of $\vec{UV_1}$ and $\vec{UV_2}$ have the same color as $U$. Precisely, the set of nodes $W$ having the same color as $U$ forms a \textit{lattice} of $\ZZ^2$ with generator vectors $\vec{UV_1}$ and $\vec{UV_2}$: the vector $\vec{UW}$ can be written as $\vec{UW}=\alpha\vec{UV_1}+\beta\vec{UV_2}$ with $\alpha$ and $\beta$ in $\ZZ^2$.\\
{\bf P4. (Coordinate-based color computation)} The grid can be seen as a tiling with the pattern $\mathcal{P_U}(\vec{UV_1},\vec{UV_2})$. Thus, each node $W$ belongs to its own parallelogram, and has coordinates relative to this parallelogram. Consequently, $W$ has the same color as any node having the same coordinates in its own parallelogram.\\

To be applied to a wireless sensor network, these principles have to be enhanced. For instance, since nodes having the same color can simultaneously access the wireless medium, validity of the coloring must be verified. Moreover, to ensure an efficient usage of the bandwidth, the number of colors used should be minimized. These criteria are taken into account in our work and progressively introduced in the paper.\\

In the following we detail the components of {\em VCM}. As previously said, we only consider grid colorings that periodically reproduce a color pattern. 

\section{VCM: Node Color Computation (NCC)}
\label{sec:vcm-color-computation}
\subsection{Assigning Colors to Nodes}


The node color computation (NCC) component of the VCM method takes as parameters two generator vectors $u_1$,$u_2$ (computed as in \refsec{sec:vcm-ovs}). Let $(x_1,y_1)$ and $(x_2,y_2)$ be their coordinates and let $d = \det(u_1,u_2)$.
Here we define two methods to compute the colors.\\
Actual computation on an example is provided in \refsec{sec:ex-ncc1-ncc2}.

\subsubsection{Method NCC1}
\label{sec:NCC1}

\mybox{
\begin{method}[VCM-NCC 1]
{\em VCM} assigns the color of a point $W$ 
based on its coordinates $w=(x,y)$ by computing first the integer quantities $c_1(w), c_2(w)$ as in System~(\ref{colorcompute}), 
\begin{System}
\label{colorcompute}
c_1(w)= (xy_2 - yx_2) \mathrm{~modulo~} d\\
c_2(w)= (-xy_1 + yx_1) \mathrm{~modulo~} d 
\end{System}
and then using a bijective mapping between the couple $(c_1,c_2)$ and a color $\in \{0, 1, 2, \ldots |d|-1 \}$.
\end{method}
}



\begin{remark}
Defining $c_2$ as $c_2(w) \triangleq (xy_1 - yx_1) \mathrm{~modulo~} d$, instead
of the definition in (\refeq{colorcompute}), would be more symmetrical compared
 to $c_1$, and would yield identical results (as it is a trivial
transformation: $c_2(w) \rightarrow (-c_2(w))\mathrm{~modulo~} d$).
\end{remark}

\begin{remark} 
\label{rem:d-positive}
In the remaining of this report, we will assume that $d > 0$ without loss of 
generality:
indeed, if $d < 0$, it is sufficient to use the vectors $(-u_1,u_2)$ instead
of $(u_1,u_2)$ and the results are similar  ; notice that 
$c_1(w,u_1,u_2) = c_1(-w,-u_1,u_2)$, etc. (hence change of sign of $u_1$ is 
equivalent to an origin symmetry of the coloring).
This avoids minor technicalities on the definition of the 
modulo, integer part, fractional part, when numbers are negative (for which
definitions are not universal).
\end{remark}

\begin{property}
\label{prop:ncc1}
With the previous coloring VCM-NCC1, it is indeed possible to define a bijection from $(c_1(w),c_2(w))$ to $\{0,1,2, \ldots |d|-1 \}$. Moreover, the coloring verifies 
principles {\bf P3} and {\bf P4} defined in Section~\ref{principlePeriodic}. 
\end{property}
\proof
\begin{LongVer}
We assume that $d>0$ without loss of generality
(see Remark~\ref{rem:d-positive}).
We need to prove that the set $\{ (c_1(w),c_2(w)) ~|~ w \in \ZZ^2 \}$ has cardinality $d$ and that the set of nodes with the same color is exactly the lattice $\Lambda(u_1,u_2)$ translated at $w$.

Let $W$ be a grid point of coordinates $w = (x,y)$. 
Performing a change of vector basis in $\mathbb{R}^2$ from $\{ (1,0), (0,1)\}$ to $\{u_1,u_2\}$, 
the new coordinates $(\alpha, \beta) \in \mathbb{R}^2$ of W in $\Lambda(u_1,u_2)$ verify $w = \alpha u_1 + \beta u_2$ and 
\begin{System}
\alpha = \frac{\det(w,u_2)}{d}\\
\beta = \frac{\det(u_1,w)}{d}
\label{eq:coord-change}
\end{System}
with $d = det(u_1,u_2)$.

Let $\alpha'$ and $\beta'$ be the integer parts of $\alpha,\beta$, i.e.,
$\alpha' = \lfloor \alpha \rfloor$ and $\beta' = \lfloor \beta \rfloor$.
For arbitrary nonzero integers $\lambda,\mu$, with also $\mu > 0$, we have the identity:
$\frac{\lambda}{\mu} = \lfloor \frac{\lambda}{\mu} \rfloor + \frac{\lambda \mathrm{~modulo~} \mu}{\mu}$.
%
%
Thus \refeq{eq:coord-change} becomes:
\begin{System} 
\alpha = \alpha'+ \frac{\det({w},u_2) \mathrm{~modulo~} d}{d} = \alpha' + \frac{c_1(w)}{d}\\
\beta = \beta'+ \frac{\det(u_1,w) \mathrm{~modulo~} d}{d} = \beta' + \frac{c_2(w)}{d}
\label{eq:coord-change2}
\end{System}
Let $W'$ the point with coordinates $w' = (\alpha', \beta')$.
$W'$ is on the lattice since $\alpha',\beta'$ are integers, 
and observe that $W$ is in fact inside the parallelogram of the lattice
$\Lambda(u_1,u_2)$ placed at node $W'$ (i.e. inside the parallelogram defined
by the $4$ points of the lattice:
$w'$, $w'+u_1$, $w'+u_2$, $w'+u_1+u_2$). Then (\refeq{eq:coord-change2})
means simply that $(\frac{c_1(w)}{d}, \frac{c_2(w)}{d})$ are the 
coordinates of $W$ relative to this parallelogram
(with the basis vectors $u_1,u_2$).
\end{LongVer}
\begin{ShortVer}
In fact, $(\frac{c_1(W)}{d}, \frac{c_2(W)}{d})$ are the coordinates of $W$ relative in the parallelogram with basis vector $(u_1,u_2)$ (see \cite{RR-VCM} for the complete proof).
\end{ShortVer}

Since there is a bijection between the set of coordinates of nodes in a parallelogram of $\Lambda(u_1,u_2)$ and the nodes themselves;
and since $(\frac{c1(w)}{d}, \frac{c2(w)}{d})$ are these coordinates, 
we have the
two properties: 1) there are exactly $d$ possible values of $(c_1,c_2)$ (because there
are exactly $d$ nodes in the parallelogram), and 2) no two nodes inside
the parallelogram have the same values $c1,c2$ since these are their coordinates, relative to one vertex of the parallelogram.
\endproof

\begin{lemma}\label{lem:colorRep}
With the color computation given by the System~\refeq{colorcompute}, the color of the node $U$ is repeated at the nodes $W$ with coordinates verifying:
$w = \alpha u_1+\beta u_2$, for some $(\alpha,\beta)$ $\in$ $\ZZ^2$.
\end{lemma}
\proof
Actually, by construction, the color of a node is given by its coordinates relative to  the parallelogram it belongs to. Hence, the color of any node $U$ is reused at nodes $W$, such that $\vec{UW}= \alpha u_1+\beta u_2$ for all $(\alpha, \beta) \in \ZZ^2$, which have the same relative coordinates.
\endproof
We deduce that {\em VCM-CC1} provides a coloring that is really consistent with the principle of the method as described in Section~\ref{principlePeriodic}.

\subsubsection{Example of bijection for NCC1}
\label{sec:example-bijection}
One of the steps of the Method NCC~1, is that a bijection needs to be established 
between the set of values $\{ (c_1(w), c_2(w)) ~|~ w \in \ZZ^2 \}$ and
the set of colors $\{0,1,\ldots, d-1 \}$: an example of bijection is provided
in this section.

A bijection can be constructed by computing the values of
$(c_1(w),c_2(w))$ for any node in $\mathcal{P}(u_1,u_2)$ in a
list, sorting the list by lexicographical order, and, finally, setting
the color associated with a couple $(c_1(w), c_2(w))$ to be its index
in the sorted list minus 1.

Example: if $(0,0)$ appears as the $1^{st}$
item of the sorted list [as it can be proved it will], 
the color assigned to that couple 
is $0$. Then, for instance, for the point $W$ of coordinates $w=(x_1,y_1)$,
in other terms $w=u_1$, we have $(c_1(w), c_2(w)) = (0,0)$ and therefore the 
color assigned to $W$ is $0$.

Note that, from a pure implementation point of view, it may be difficult to 
enumerate exactly the points of 
$\mathcal{P}(u_1,u_2)$, but then, instead, it is sufficient to enumerate all the
nodes in a superset, the \emph{bounding box} of $\mathcal{P}(u_1,u_2)$, 
itself computed from 
its four vertices $O = (0,0)$, $O+u_1$, $O+u_2$ and $O+u_1+u_2$. Computing
the set of values $(c_1(w), c_2(w))$ for the points in the bounding box,
will yield all possible values for any $w \in \ZZ^2$.

\subsubsection{Method NCC2}

Method 2 is derived from the first method ; the difference is that it
also establishes a direct bijection (and hence avoids the need for constructing
a bijection as in section~\ref{sec:example-bijection}):
it proceeds to a direct
computation of the color based on the node coordinates without computing the
colors of the other nodes on the parallelogram.

As for Method 1, we assume that we are given two generator vectors $u_1, u_2$
with coordinates $(x_1, y_1)$ and $(x_2, y_2)$, and that we are computing
the color of node of coordinates $w = (x,y)$. 

In addition to the notations and definitions used previously, we introduce
the following ones:
\begin{itemize}
\item $g_1$ is the greatest common divisor of $(x_1,y_1)$; $g_1=gcd(x_1,y_1)$.
Similarly, $g_2=gcd(x_2,y_2)$, with the convention $gcd(a,0) = gcd(0,a) = a$
\item $v_1$ is the vector $=\frac{1}{g_1}u_1$ and $v_2$ is the vector $=\frac{1}{g_2}u_2$. Notice that the coordinates of $v_1$ and $v_2$ are coprimes.
\item Let $d'= det(v_1,v_2)$. We have: $d'=\frac{d}{g_1g_2}$
\item Let $(\alpha'(w),\beta'(w)) \in \RR^2$ be the coordinates of the node $w$ relative to the basis $(v_1,v_2)$, when performing a change of basis in $\RR^2$ from $(1,0), (0,1)$ to $(v_1,v_2)$. We have: $\alpha'(w) =\frac{det(w,v_2)}{d'}$ and $\beta'(w) =\frac{det(v_1, w)}{d'}$.
\item Let $x'(w) = \lfloor \alpha'(w)\rfloor \mathrm{~modulo~} g_1$  
\item Let $y'(w) = \lfloor \beta'(w)\rfloor \mathrm{~modulo~} g_2$
\item Let $c'(w)= det(w,v_2)  \mathrm{~modulo~} d'$
\end{itemize}
~\\
\mybox{
\begin{method}[VCM-NCC 2]
Using the notations mentioned above, a color of a node with coordinates 
$w=(x,y)$ is equal to 
\begin{equation}
c(w) = c'(w) + d' x'(w)+ d'g_1y'(w)
\end{equation}
which is an integer in $\{0,1, 2, \ldots |d|-1 \}$
\end{method}
}

As in remark~\ref{rem:d-positive}, we now assume that $d>0$ without loss of generality.

The idea of the method VCM-NCC 2 is as follows:
\begin{itemize}
\item In method VCM-NCC 1,
we did not identify nor did make use of the special structure of 
$E_c = \{ (c_1(w), c_2(w)) ~|~ w \in \ZZ^2 \}$,  which is in fact a subgroup of
$(\frac{\ZZ}{d \ZZ})^2$ ; hence we did not provide an explicit mapping from 
$E_c$
to $\{0, 1, 2, \ldots d - 1 \}$. This is the problem that method 2 addresses:
\item In method VCM-NCC 2,
by dividing
the vectors $u_1$ and $u_2$ by the respective gcd of their coordinates,
we obtain vectors $v_1, v_2$, and we find that applying method VCM-NCC 1
with these vectors, the set of values $\{ c_1(w, v_1, v_2), c_2(w,v_1,v_2)
~|~ w \in \ZZ^2 \}$ has good properties, allowing the construction of a direct
bijection (see lemma~\ref{lem:c-prime-is-a-direct-coloring}).
\item However using method VCM-NCC 1 with vectors $v_1,v_2$ yields a
coloring of $\ZZ^2$ repeated by translation by $v_1$ and $v_2$ ;
whereas we wanted a coloring repeated by translation by $u_1$ and $u_2$
(which are larger). Notice that the problem here is that method VCC-NCC 1 with
generator vectors $v_1,v_2$ is quite possibly an invalid $h$-hop coloring
of $\ZZ^2$, even if $u_1, u_2$ are vectors defining a valid $h$-hop coloring.
For this reason, method VCM-NCC 2 constructs a coloring based on
method VCM-NCC 1, but modified, by tiling the parallelogram
$\mathcal{P}(v_1,v_2)$ several times inside the parallelogram
$\mathcal{P}(u_1,u_2)$ and changing the colors in each internal tile 
(see figure~\ref{fig:NCC2}).
\end{itemize}

\begin{figure}[!h]
\begin{center}
{\includegraphics[width=0.99\linewidth]{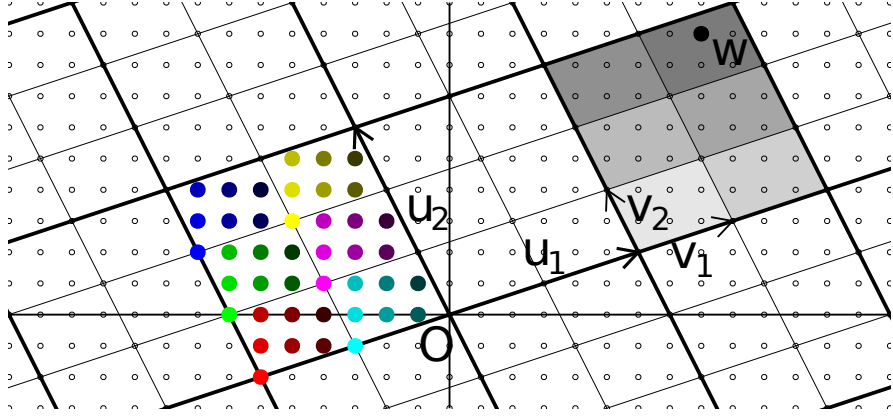}}
\caption{Method VCM-NCC 2}\label{fig:NCC2}
\end{center}
\end{figure}

Method 2 is in fact the combination of two functions; $c(w) = b_3(c_3(w))$ with:
\begin{itemize}
\item The first one, $c_3$, from 
$\ZZ^2 $ to $ \{0, 1, .. d'-1 \} \times \{0, 1, .. g_1-1 \} \times 
  \{0, 1, .. g_2-1 \} $, which associates to $w$
the value  $c_3(w) = (c'(w), x'(w), y'(w))$.
\item
 The second is bijection $b_3$
from $\{0, 1, \ldots d'-1 \} \times \{0, 1, \ldots g_1-1 \} \times 
\{0, 1, \ldots g_2-1 \} $
to $\{ 1, 2, \ldots d \}$ by transforming $(j,k,l) \rightarrow b_3(j,k,l) =  j + d' k + d' g_1 l $  ; note that $d = g_1g_2d'$ by definition.
\end{itemize}

The transformation $b_3$ is obviously a bijection from its domain to its 
codomain\footnote{informally,
it is similar to a
 transforming the time of the day expressed as h:m:s -- $h$ hour $m$ minutes $s$ seconds -- into the time of the day expressed as seconds from midnight as $(h,m,s) \rightarrow 3600 \times h + 60 \times m + s$}, and therefore the crux of method VCM-NCC 2 is the function $c_3$.

We start with the following lemma for the function $c'(w)$, which is
in fact $c_1(w,v_1,v_2)$:
\begin{lemma}
\label{lem:c-prime-is-a-direct-coloring}
Let $c_1(w,v_1,v_2)$ be the value ``$c_1(w)$'' obtained when applying method
VCM-NCC~1 with generator vectors $(v_1,v_2)$, and define $c_2(w,v_1,v_2)$
similarly.\\
Then the function $w \rightarrow c_1(w,v_1,v_2)$ is a direct coloring of the 
nodes $w$, where the color is in $\{0, 1, \ldots d'-1 \}$.
\end{lemma}
This lemma establishes that for the specific generator vectors $(v_1,v_2)$,
$c_1$ computes directly a color and $c_2$ is essentially redundant.\\
\proof
We prove this lemma in two steps: \textbf{(S1)} proves that the set $\{ c_1(w,v_1,v_2) ~|~ w \in \ZZ^2 \}$ covers exactly integers in $\{ 0, 1, \ldots d'-1 \}$ , and \textbf{(S2)} the coloring provided by $c_1$ has the same properties as $VCM-NCC1$. To prove this, we prove that there is a bijection between $\{ c_1(w,v_1,v_2) ~|~ w \in \ZZ^2 \}$ and $E_c = \{ (c_1(w,v_1,v_2), c_2(w,v_1,v_2)) ~|~ w \in \ZZ^2 \}$.\\

\textbf{(S1)} Consider the vector $u_2$. By definition, $g_2$ is the gcd of its coordinates
$(x_2,y_2)$. Applying B\'ezout's identity (\cite{bib:bezout}), there
exist integers\footnote{this is even true when $x_2=0$ or $y_2=0$ with our convention $gcd(a,0)=gcd(0,a)=a$} $(\lambda,\mu) \in \ZZ^2$, such that:
\begin{equation}
\label{eq:bezout}
\lambda x_2 + \mu y_2 = g_2
\end{equation}

Let $w_0$ be the vector of coordinates $(\mu, -\lambda)$ and let
$(x_2',y_2')$ be the coordinates of $v_2$ (recall that $v_2$
is defined as $v_2 = \frac{1}{g_2}u_2$).
Because $v_2 = \frac{1}{g_2}u_2$, we have $x_2' = \frac{x_2}{g_2}$
and $y_2' = \frac{y_2}{g_2}$ and hence equation (\ref{eq:bezout}) becomes:
$$\lambda x_2' + \mu y_2' = 1$$
$$\implies det(w_0,v_2) = 1$$  
$$\implies c_1(w_0,v_1,v_2) = 1$$

Since $c_1$ is a linear map:
 $c_1(2 w_0,v_1,v_2) = 2$, etc. and by applying it to $0, w_0, 2 w_0, \ldots (d'-1) w_0$, we get
exactly:
$$\{ c_1(w,v_1,v_2) ~|~ w \in \ZZ^2 \} = \{ 0, 1, \ldots d'-1 \}$$.

\textbf{(S2)} Let $E_c$ be the set
$E_c = \{ (c_1(w,v_1,v_2), c_2(w,v_1,v_2)) ~|~ w \in \ZZ^2 \}$.
From properties previously proven for method VCM-NCC~1 
(including property~\ref{prop:ncc1}), we know that there exists
a bijection from $E_c$ to the set $\{ 0, 1, 2, \ldots d'-1 \}$ ; in other words,
that its cardinality verifies $|E_c| = d'$.

Now consider the projection $p_x : (x,y) \rightarrow x$. The previous
relation shows that $p_x(E_c) = \{ 0, 1, \ldots d'-1\}$ hence it is a bijection
because it is surjective between two sets of cardinality $d'$. 

Therefore we have explicitly found one bijection 
(precisely: $k,l \rightarrow p_x(k,l)$) compared
to the method VCM-NCC~1 where we only 
established the existence of such a bijection. 
The end result, is that when applying
VCM-NCC~1 with the generator vectors $(v_1,v_2)$ and with this bijection $p_x$,
the color of a node $w$ is $c_1(w,v_1,v_2)$. The coloring provided by $c_1$ inherits from VCM-NCC~1, all the properties listed in \refsec{principlePeriodic}.
\endproof

~\newline
The next step is to build a coloring repeated by translation of $u_1$ and $u_2$,
from the coloring VCM-NCC~1. As illustrated in Figure~\ref{fig:NCC2}, the
idea is that $v_1,v_2$ are the basis of the parallelogram $\mathcal{P}(v_1,v_2)$
that can be tiled to cover the parallelogram $\mathcal{P}(u_1,u_2)$ ; this
is possible because $u_1 = g_1 v_1$ and $u_2 = g_2 v_2$ (where $g_1$ and $g_2$
are integers). Inside a parallelogram $\mathcal{P}(v_1,v_2)$, the nodes
are assigned the same color as given by VCM-NCC~1 with $v_1,v_2$, except that
an offset is added, depending in which tile the node is located. 

Concretely,
in Figure~\ref{fig:NCC2}, the parallelogram $\mathcal{P}_{u_1}(u_1,u_2)$ at the
right side of the picture is tiled with $6$ smaller parallelograms which are
versions of $\mathcal{P}(v_1,v_2)$, filled by different shades of gray.
The nodes in one small gray parallelogram are allocated the colors
$0, 1, 2 \ldots d'-1$ ; the ones of the next small gray parallelogram are 
allocated the colors $d', d'+1, \ldots 2d'-1$; and generally the nodes
of the $k^{th}$ parallelogram are allocated the colors 
$kd', kd'+1, \ldots (k+1)d'-1$. The actual coloring is then something like the colored nodes at the left of the Figure~\ref{fig:NCC2}.

Thus, the central idea is to identify in which sub-parallelogram a node
is located: this is given by $x'(w),y'(w)$, which are in fact
the coordinates of the smaller parallelogram on the basis $(v_1,v_2)$
counted relative to the larger parallelogram.
For instance, in Figure~\ref{fig:NCC2}, the point $w$ is inside a
sub-parallelogram of coordinates $(x'(w), y'(w)) = (1,2)$ relative
to the bigger parallelogram.

This is the informal idea behind the quantities defined previously: in the
following lemma, we now prove formally that VCM-NCC~2 is a coloring satisfying
the desired properties.

\begin{lemma}
The coloring provided by the method VCM-NCC~2 satisfies the principles
{\bf P1}, {\bf{}P2}, {\bf{}P3} and {\bf{}P4} from \refsec{principlePeriodic}.
\end{lemma}
\proof

Almost all the properties can be derived from identifying the set of nodes
with identical colors.

Let $w_1$ and $w_2$ be two nodes with the same color, hence, 
with $c(w_1)=c(w_2)$.
Because $c = b_3 \circ c_3$ and $b_3$ is a bijection, this implies
that 
\begin{equation}
\label{eq:w1w2equal}
c'(w_1) = c'(w_2) \mathrm{~and~} x'(w_1) = x'(w_2) \mathrm{~and~} y'(w_1) = y'(w_2)
\end{equation}

The part $c'(w_1) = c'(w_2)$ implies that $w_1$ and $w_2$ have the same color
in the method VCM-NCC~1 applied with vectors $v_1,v_2$, therefore, by
properties of this method:
$$w_1-w_2 = \lambda v_1 + \mu v_2 \mathrm{~for~some~}(\lambda,\mu) \in \ZZ^2$$

Let us decompose $\lambda$ as quotient and remainder modulo $g_1$
(resp. $\mu$, modulo $g_2$):
$\lambda = g_1q_1 + r_1$ and $\mu = g_2q_2 + r_2$ where
$r_1 \in \{ 0, 1, \ldots g_1 -1\}$, $r_2\in\{0, 1, \ldots g_2 -1\}$
and $(q_1,q_2)\in\ZZ^2$.

The previous equality becomes:
\begin{equation}
\label{eq:diffw1w2}
w_1-w_2 = (g_1 q_1 + r_1) v_1 + (g_2 q_2 + r_2) v_2
\end{equation}

Now developing the part $x'(w_1) = x'(w_2)$ from (\ref{eq:w1w2equal}),
we get:
$$\lfloor \alpha'(w_1)\rfloor \mathrm{~modulo~} g_1 = \lfloor \alpha'(w_2)\rfloor \mathrm{~modulo~} g_1$$
$$\implies \lfloor \alpha'(w_1) - \alpha'(w_2)\rfloor \mathrm{~modulo~} g_1 = 0
\mathrm{~or~}
\lfloor \alpha'(w_2) - \alpha'(w_1)\rfloor \mathrm{~modulo~} g_1 = 0$$
Which equality is true depends on which one of $\alpha'(w_1)$ and
$\alpha'(w_2)$ has the largest fractional part. Assume it is $\alpha'(w_1)$,
then we have:
\begin{eqnarray*}
& \lfloor \alpha'(w_1) - \alpha'(w_2)\rfloor \mathrm{~modulo~} g_1 = 0 \\
\implies & \lfloor \frac{det(w_1-w_2,v_2)}{d'}\rfloor \mathrm{~modulo~} g_1 = 0~\mathrm{~(by~definition~of~\alpha')} \\
\implies & \lfloor (g_1q_1 + r_1) \frac{det(v_1,v_2)}{d'}\rfloor \mathrm{~modulo~} g_1 = 0~\mathrm{~(using~eq~\ref{eq:diffw1w2})} \\
\implies & r_1 = 0~\mathrm{~(using~the~definition~of~d'})
\end{eqnarray*}

In the same way, developing $y'(w_1)=y'(w_2)$ from (\ref{eq:w1w2equal}), 
yields $r_2 = 0$. Since $u_1=g_1v_1, u_2=g_2v_2$, we can now rewrite 
Eq~(\ref{eq:diffw1w2}) as:
$$w_1-w_2 = q_1 u_1 + q_2 u_2~\mathrm{~for~some~}(q_1,q_2)\in\ZZ^2$$

The opposite is true: if $w_1-w_2$ verify this equality,
$w_1 = w_2 + q_1 u_1 + q_2 u_2$, and then one can easily check that 
$c'(w_1) = c'(w_2)$ and similarly:
\begin{eqnarray*}
x'(w_1)&=&\lfloor \alpha'(w_2 + q_1 u_1 + q_2 u_2) \rfloor 
            \mathrm{~modulo~} g_1 \\
    &=&\lfloor \frac{det(w_2 + q_1 u_1 + q_2 u_2, v_2)}{det(v_1,v_2)} \rfloor 
             \mathrm{~modulo~} g_1 \\
    &=&\lfloor \alpha'(w_2) + q_1 g_1 \rfloor 
             \mathrm{~modulo~} g_1 \\
    &=&\lfloor \alpha'(w_2)  \rfloor 
             \mathrm{~modulo~} g_1 = x'(w_2)
\end{eqnarray*}
Likewise $y'(w_1) = y'(w_2)$ and thus finally $c(w_1)=c(w_2)$, that is, the
two nodes of coordinates $w_1$ and $w_2$ have the same color.

As a result, we have established that the nodes with identical colors are 
exactly the points located on a lattice generated by vectors $v_1,v_2$. 
This is exactly the principle {\bf{}P3}, and actually implies the principles
{\bf{}P1}, {\bf{}P2} and {\bf{}P4}.

\endproof

%

 


\begin{remark}
In this section, we selected $c'(w)$ as $c'(w) = c_1(w,v_1,v_2)$. Alternatively,
one could select $c'(w)=c_2(w,v_1,v_2)$. Notice also that once the choice
is made, lemma~\ref{lem:c-prime-is-a-direct-coloring} does not require
that coordinates are coprimes for both vectors $v_1$ and $v_2$. Indeed, for the choice $c' = c_1$, (respectively $c' = c_2$) only $v_2$ (respectively $v_1$) is required to have coordinates that are coprime.
\end{remark}

\subsection{Computing the Number of Colors}
The number of colors used in a periodic $h$-hop coloring is given by the next Property.
\label{sec:color-compute}
\begin{property}
\label{numberofcolors}
For any node $U$, the color pattern defined by the two generator vectors $u_1$ and $u_2$ meeting the aforementioned principles contains exactly
$| x_1y_2-x_2y_1 |$ colors where $(x_1,y_1)$ and $(x_2,y_2)$ are the
coordinates of $u_1$ and $u_2$.
\end{property}
\proof By definition, no two nodes within the parallelogram defined by $u_1$ and $u_2$ use the same color. Hence the number of colors is equal to the number of nodes in this parallelogram. Moreover, as we said, the number of nodes in $\mathcal{P}(u_1,u_2)$, called \textit{lattice determinant}, is equal to the absolute value of det$(u_1,u_2)$~\cite{latticeDef}. Hence the property.
\endproof

\subsection{Example of Color Calculation}
\label{sec:ex-ncc1-ncc2}
In this section, we illustrate the color calculation, using the example
of Figure~\ref{fig:NCC2}.

In Figure~\ref{fig:NCC2}, we have the following coordinates for $u_1,u_2,w$:
\begin{itemize}
\item $u_1 = (6,2)$ and $u_2 = (-3,6)$
\item $w = (8,9)$
\end{itemize}

Applying VCM-NCC~1 with the generator vectors $u_1,u_2$, we get:
\begin{itemize}
\item Number of colors $= 42$
\item $(c_1(w,u_1,u_2), c_2(w,u_1,u_2)) = (33, 38)$
\item Using the example bijection of \refsec{sec:example-bijection},
  $c_1(w),c_2(w)$ is the $36^{th}$ value in the sorted list of possible values,
  hence color$(w) = 35$ 
\end{itemize}

Applying VCM-NCC~2 with the generator vectors $u_1,u_2$, we get:
\begin{itemize}
\item Number of colors $= 42$
\item $g_1 = gcd(\mathrm{coords~of~} u_1) = 2$ and then $v_1 = (3,1)$
\item $g_2 = gcd(\mathrm{coords~of~} u_2) = 3$ and then $v_2 = (-1,2)$
\item $c'(w) = 4$, $x'(w)=1$, $y'(w)=2$
\item Color of $w$: $c(w) = 39$
\end{itemize}

For reference, the colors are computed internally from the coordinates $w$ 
on different basis:
$w = (1+\frac{33}{42}) u_1 + \frac{38}{42} u_2$
and $w = (3+\frac{4}{7}) v_1 + (2+\frac{5}{7}) v_2$

\section{VCM: Validity Check (VC)}
\label{sec:vcm-validity}\label{sec:vcm-vc}
As defined previously, a $h$-hop coloring algorithm is valid if and only if no two nodes that are at less or equal to $h$-hop from each other use the same color. 


The node color computation algorithm of VCM 
(described in Section~\ref{sec:vcm-color-computation}) takes as input two generator
vectors $u_1$, and $u_2$, and gives the color of each node. 
In this section, we will assume that such
two vectors are given and fixed, and we present two methods for checking
beforehand whether the coloring induced by these vectors is a valid coloring.


\subsection{Method VC1: Verification around Origin}
\label{sec:vcm-vc-m1}

\mybox{
\begin{method}[VC1]
For each node $W$ in the $h$-hop neighborhood of the origin node $O$, we compute the color of this node based on the given generator vectors $u_1$ and $u_2$. If $W$ has the same color as $O$, then we conclude that the vectors $u_1$ and $u_2$ do not provide a valid coloring. Otherwise, if the color of $O$, is not repeated at any point $W$ in its $h$-hop neighborhood, then the coloring is valid.
\end{method}
}
The idea of Method~\ref{lem:origin-check-equiv} is based on the following fact, proven in this section: if there is a color conflict between any two nodes $V_1$ and $V_2$ in $\Lambda(u_1,u_2)$ ($V_1$ and $V_2$ have the same color despite they are at less than or equal to $h$ hops), there will be a color conflict in the $h$-hop neighborhood of the origin $O$. \\
We set $d = \det(u_1,u_2)$.

\begin{lemma} \label{lem:origin-check-equiv}
If two nodes $V_1$ and $V_2$ with coordinates $v_1$, $v_2$ in $\ZZ^2$ have the same color, then the color of the origin node is repeated at the node $W$ of coordinates $v_2 - v_1$.
\end{lemma}
\proof
\begin{LongVer}
The functions $c_1,c_2$ computed from System~\refeq{colorcompute} are actually linear modulo $d$. That is, if $W$ is the node with coordinates $v_1-v_2$, and $w$ is the vector of nodes extremities the origin and $W$, we get:
$c_1(W) = c_1(V_1)-c_1(V_2) \mathrm{~modulo~} d$. Hence, if $c_1(V_1) = c_1(V_2)$ we have $c_1(W) = c_1(O)$. This is true also
for $c_2$, hence the lemma.
\end{LongVer}
\begin{ShortVer}
The property comes from the fact that $c_1$ and $c_2$ are linear (see \cite{RR-VCM} for details).
\end{ShortVer}
\endproof

\begin{theorem}
\label{method:disk-zero-check}
If none of the nodes inside the $h$-hop neighborhood of the origin node 
$O=(0,0)$
has the same color as $O$ itself, then the coloring is valid.
\end{theorem}
\proof
By contradiction: assume that the coloring is invalid, which implies that
two nodes $V_1,V_2$ at less or equal to $h$ hops have the same color. Then from
Lemma~\ref{lem:origin-check-equiv}, the node $W$ such as $\vec{OW} = \vec{V_1V_2}$ has the same color as $O$. Notice that the distance in terms of hop number between $O$ and $W$ is the same as the distance between $V_1$ and $V_2$. Hence we have found a color conflict between $O$ and a node $W$ which is at less than $h$ hops from $O$. Hence the theorem.
\endproof

Theorem~\ref{method:disk-zero-check} proves that Method VC1 is a correct method
for checking whether two generator vectors yield a valid $h$-hop coloring.

\subsection{Method VC2: Verification in a Few Points}
\label{sec:vcm-vc-m2}

Method VC\ref{method:disk-zero-check} requires $\Theta(R^2)$ 
verifications when $R \rightarrow \infty$. In the following, we propose Method VC\ref{method:small-grid-check}, usable when $R>\sqrt{2}$ and requiring only a bounded number of verifications. Method VC\ref{method:small-grid-check} performs a check on a few nodes on the lattice $\Lambda(u_1,u_2)$ to guarantee that the $h$-hop coloring associated to $u_1, u_2$ is valid. This method is based on Gauss lattice reduction~\cite{VV07} (see the Annex for more details): $u_1$ and $u_2$ should be first reduced,
and hence verify the Equations~(\ref{GaussPr}). \\
\mybox{
\begin{method}[VC2]
The nodes with the same color as the origin are on the lattice 
$\Lambda(u_1,u_2)$:
this method verifies that these nodes are at least $(h+1)$-away from $O$, in which case the coloring is valid. 
However, not all grid nodes need to be checked. It is sufficient to check only nodes $W$ in $\Lambda(u_1,u_2)$ with coordinates $\alpha,\beta$ on the basis $\{ u_1, u_2 \}$, such that $|\alpha|$ and $|\beta| < \mu(R)$, with $\mu(R) = \frac{2 \sqrt{3}R}{3(R-\sqrt{2})}$. 
The coloring is valid if and only if these nodes are strictly more than $h$ hops from the origin node.
\end{method}
}
This method is based on the following theorem.

\begin{theorem}
\label{method:small-grid-check}
For $R>\sqrt{2}$, the coloring provided by two reduced vectors $u_1, u_2$ is valid if and only if: \newline 
for all $\alpha,\beta$ integers verifying $|\alpha| < \mu(R)$, 
and $|\beta| < \mu(R)$, 
the node with coordinates $(\alpha,\beta)$ on the basis $\{ u_1, u_2 \}$ 
is at strictly more
than $h$ hops from the origin node $O$, where $\mu(R) = \frac{2\sqrt{3}R}{3(R-\sqrt{2})}$.
\end{theorem}
\proof 
\begin{LongVer}
The property comes from the fact that the points on the lattice $\Lambda(u_1,u_2)$
are ``far'' from the origin node, because the vectors $u_1,u_2$ are reduced.\\
Indeed, Lemma~\ref{lem:mu-bound} (see the Annex) means that any node on the lattice with coordinates $\alpha u_1 + \beta u_2$, with $|\alpha|$ or $|\beta|$ $\ge \mu(R)$ can reuse the color of the origin node $O$ because they are at strictly more that $h$-hop from $O$ (provided that the points of coordinates $u_1$ or $u_2$ are themselves strictly more than $h$-hop away from $O$). Hence, to check the validity of the coloring provided by {\em VCM}, it is necessary and sufficient to check that for all $|\alpha|$ and $|\beta| < \mu(R)$, nodes of coordinates $\alpha$, $\beta$ in the lattice $\Lambda(u_1,u_2)$ are strictly more than $h$ hops away from the origin of the lattice.
This check includes checking the validity of $u_1$ and $u_2$ themselves 
(cases $(\alpha,\beta) = (1,0)$ and $(\alpha,\beta) = (0,1)$)

~\newline
Notice that for dense grids ($R\rightarrow \infty$), 
$\mu \rightarrow 1.15\ldots $. This small bound reduces the set of nodes to be checked in order to verify the validity of the coloring for given vector candidates.

In fact for $R > \frac{3 \sqrt{2}}{3-\sqrt{3}}$,
that is for $R > 3.3461$, we have $\mu < 2$, hence only $4$ points need to
be checked (considering symmetries):
$u_1$ (with $\alpha = 1, \beta=0$), $u_2$ (with $\alpha = 0, \beta = 1$),
$u_1+u_2$ (with $\alpha = 1, \beta = 1$),
$u_1-u_2$ (with $\alpha = 1, \beta = -1$).

However, Method~\ref{method:disk-zero-check} is applicable for any radio range $R$, whereas Method\ref{method:small-grid-check} requires $(R>\sqrt{2}).$

\end{LongVer}

\begin{ShortVer}
The property comes from the fact the points on the lattice $\Lambda(u_1,u_2)$
are ``far'' from the origin node, because the vector $u_1,u_2$ are reduced.
The details of the proof are available in \cite{RR-VCM}.
\end{ShortVer}

\section{VCM: Optimal Vector Search (OVS)}
\label{sec:vcm-ovs}
To achieve an optimal spatial reuse, the coloring algorithm should minimize the number of colors used.
For \textit{VCM}, our aim is to judiciously choose the generator vectors $u_1$ and $u_2$ in order to reduce the number of colors used to color a grid. However, by default, the infinite lattice $\ZZ^2$ is a possible set for candidate vectors. So, to find the optimal vectors in  a small set, our approach is as the following:
\begin{itemize}
\item We determine the upper and lower bounds on the number of colors needed in a $h$-hop coloring of the grid, for $h\ge1$. 

\item Because for any couple of initial generator
vectors, reduced vectors always exist, it is sufficient to search 
for some optimal vectors in the space of reduced vectors defined
by System~(\ref{GaussPr}).
\end{itemize}

We will show in Section~\ref{sec:complexity} how to bound the set of candidate vectors using the properties of {\em lattice reduction} and the  upper and lower bounds on the number of colors to decrease
the complexity of the search for the optimal vectors.
\newline

\subsection{Bounds on the Number of Colors in Colorings}\label{boundsSec}
In this section, we prove that the number of colors of optimal colorings when $R \rightarrow \infty$ is shown to be asymptotically
$\frac{\sqrt{3}}{2}h^2R^2 + O(R)$, from the combination of two bounds.

\subsubsection{Lower Bound}
For the lower bound, we have the following theorem that is valid for any coloring, not just periodic colorings. It uses known results on circle packings.

\begin{theorem}\label{lowerBoundTheorem}
The number of colors required to color an infinite grid with $R>\sqrt{2}$ is at least $\frac{\sqrt{3}}{2}h^2(R-\sqrt{2})^2$.
\end{theorem}

\begin{proof}
Consider $h$-hop coloring of the grid $\ZZ^2$. 
Consider a fixed color $c$, and now let $S_c$ be the set of nodes having
this color. 

We first establish a lower bound of the distance of nodes
in $S_c$. 
Let us define $\rho = (R-\sqrt{2})h$.
Consider two nodes $A,B$ of $S_c$. By contradiction: if their distance
verifies $d(A,B) \leq \rho$, from Lemma~\ref{hhopprop},
they would be at most $h$-hop away, contradicting the definition of a 
$h$-hop coloring. 
Therefore, all nodes of $S_c$ are at a distance at least $\rho$ from
each other. 

Now consider the set of circles ${\cal C}$ of 
radius $\frac{1}{2}\rho$
and whose centers are the nodes of $S_c$. The fact that any two nodes 
of $S_c$ are distant of more $\rho$, implies that none of the circles 
in ${\cal C}$ overlap. Hence ${\cal C}$ is a \emph{circle packing} 
by definition. From the Thue-T\'oth theorem 
\cite{t1910,t1943} establishing that the hexagonal circle packing 
is the densest packing, with a density of $\frac{\pi}{\sqrt{12}}$, 
we deduce that ${\cal C}$
must have a lower or equal packing density. 
This implies an upper bound of the density of set $S_c$ of centers of the disks of $\frac{1}{(\rho/2)^2\sqrt{12}}$.

Because each color yields a set of nodes with at most this density, it
follows a lower bound of the number of colors that is the inverse of this
quantity, hence the theorem.
\end{proof}

\subsubsection{Upper Bound}
For an upper bound, when $R > \sqrt{2}$ , we construct explicitly a periodic 
coloring; more 
precisely we construct two vectors $v_1, v_2$ yielding a valid coloring
with VCM. As a result, optimal colorings and optimal periodic colorings
must have a number of colors which is lower or equal.

Because the vectors constructed yield lattices which are close to hexagonal
lattices (when $R \rightarrow \infty$), and because it happens that hexagonal
lattices yields the densest packing (as used in the proof of
theorem~\ref{lowerBoundTheorem}),
the upper bound will be somewhat ``close'' to the previous lower bound.

We have the following theorem:
\begin{theorem}\label{theo:upper-bound}
The number of colors required to color an infinite grid is at most
$\frac{\sqrt{3}}{2}h^2R^2 + 2 hR + (hR+2) \sqrt{2}$.
\end{theorem}

\begin{proof}
We proceed with a constructive proof, exhibiting two valid vectors which yield 
the result, using an approximation of an hexagonal lattice.

Figure~\ref{fig:hexagon-choice} illustrates how some points $V_1$ and
$V_2$ are constructed. 
\begin{figure}[!h]
\begin{center}
{\includegraphics[width=0.6\linewidth]{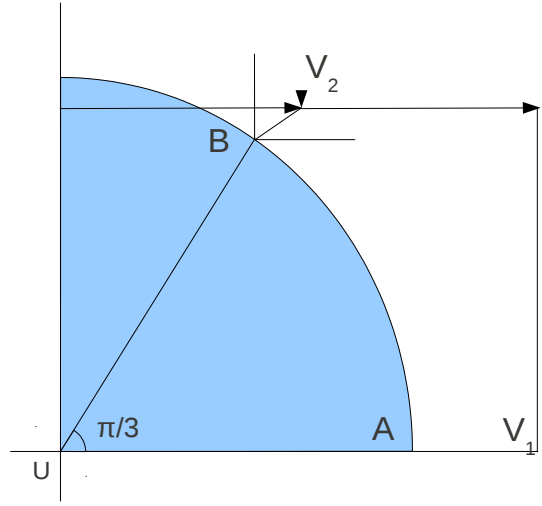}}
\caption{Selecting vectors for a near-hexagonal lattice}\label{fig:hexagon-choice}
\end{center}
\end{figure}
\begin{itemize}
\item Starting from the point $U$, the line with an 
angle $\frac{\pi}{3}$ with the horizontal line is considered,
and its intersection with the circle of radius $h R$ yields the point $B$.
\item Next, the closest point of $B$ on the grid
  having higher coordinates $x$ and $y$ than $B$ is sought and is $V_2$ with coordinates $(x_2, y_2)$.
\item Then $V_1$ with coordinates $(x_1, y_1)$ is selected with
  $(x_1,y_1) = (2x_2, 0)$.
\end{itemize}

Notice that by construction $x_1 \ge h R$, and we have a valid choice of
vectors $\vec{UV_1}$ and $\vec{UV_2}$.

\begin{LongVer}
Denote $(\gamma,\delta)$ the coordinates of $\vec{BV_2}$.
By construction: $0 \le \gamma \le 1$ and $0 \le \delta \le 1$, hence:
$|\vec{BV_2}|  \le \sqrt{2}$. Moreover, 
\begin{eqnarray*}
|\vec{AV_1}| &=& 2x_2 - hR\\
			 &=& 2(hR \cos(\frac{\pi}{3}) + \gamma) -hR, \mathrm{~with~ }\gamma \le 1\\
			 &\le& 2
\
\end{eqnarray*}

Consequently, we can write $n_c$, the number of colors in the associated coloring as:
\begin{small}
\begin{eqnarray*}
n_c & = & det(\vec{UV_1},\vec{UV_2}) \\
    & = & det(\vec{UA},\vec{UB}) + det(\vec{AV_1},\vec{UB}) + det(\vec{UV_1},\vec{BV_2}) \\
    & \le & det(\vec{UA},\vec{UB}) + |\vec{AV_1}||\vec{UB}| + |\vec{UV_1}||\vec{BV_2}| \\
    & \le & \frac{\sqrt{3}}{2}h^2R^2 + 2 hR + (hR+2) \sqrt{2}.
\end{eqnarray*}
\end{small}
\end{LongVer}
\begin{ShortVer}
Writing the lattice determinant, we see that $n_c  \le  \frac{\sqrt{3}}{2}h^2R^2 + 2 hR + (2+hR) \sqrt{2}$.
\end{ShortVer}
\end{proof}

\begin{remark}
\label{rem:near-square-vector}
An alternate, simpler, choice of vectors is to compute the integer 
$\lambda = \lfloor hR \rfloor + 1$, and select the vectors with coordinates
$u_1' = (\lambda, 0)$ and $u_2' = (0, \lambda)$. The number of colors is
higher than for the near-hexagonal previous vectors 
(it is $h^2R^2 (1+O(\frac{1}{R}))$), hence the vectors cannot yield the result 
of the next section, but the vectors can be used for an upper bound
when $R \le \sqrt{2}$.
\end{remark}

\subsubsection{Asymptotic Number of Colors}

\begin{theorem}\label{theo:opt-color-num} 
The number of colors $n_c(R) $ of an optimal periodic $h$-hop coloring  for a fixed $h$
verifies:
$$n_c(R) = \frac{\sqrt{3}}{2}h^2R^2 (1 + O(\frac{1}{R}))$$
when $R \rightarrow \infty$.
\label{th:asymptotic}
\end{theorem}
\begin{proof}
Combining the lower bound and the upper bound of the two theorems~\ref{lowerBoundTheorem} and~\ref{theo:upper-bound}
yields the result.
\end{proof}

\begin{corollary}
Theorem~\ref{theo:opt-color-num} is true even considering periodic and non periodic colorings.
\end{corollary}
\proof
A periodic coloring is a special case of general coloring (including periodic and non periodic). Hence, the optimal number of colors in general coloring is less than or equal to $n_c$.
\endproof
\begin{corollary}
{\em VCM} is asymptotically optimal even considering all possible valid colorings (even non-periodic). In other terms {\em VCM} is
an $(1+g(R))$-approximation of the optimal coloring(s) of the grid with some $g$ verifying $g(R) \rightarrow 0$ (precisely: $g(R) = O(\frac{1}{R})$) 
when $R \rightarrow \infty$.
\end{corollary}
\begin{proof}
{\em VCM} will find vectors with better or equal performance than those in Theorem~\ref{theo:upper-bound}. Indeed, in the worst case, these generator vectors for the ``near-hexagonal'' lattice will be selected by {\em VCM}.
\end{proof}

\subsection{Finding Optimal Vectors}\label{sec:complexity}
The complexity of {\em VCM} lies in the generator vectors computation and in their validity check. In this section, we show how to limit the set of candidate vectors.\\
Let $u_1,u_2$ be two candidate vectors and $\theta$ be the angle between them. We search ${l_1}_{min}, {l_1}_{max}$ (respectively ${l_2}_{min},{l_2}_{max}$) the lower and upper bounds of the length of $u_1$ (respectively $u_2$).\\
1. Considering a $h-hop$ coloring, the vectors $u_1$ and $u_2$ must be valid. 
According to Lemma~\ref{hhopprop}, we have:  
\begin{equation}\label{minLengthEqu2}
|u_1| > (R - \sqrt{2})\ h \mathrm{~and~} |u_2| > (R-\sqrt{2}) \ h. \\
\end{equation}
2. As we said, in order to reduce the set of candidate generator vectors, we reduce the size of these vectors by using the lattice reduction algorithm of Gauss. A consequence of Gauss property~\ref{GaussPr} is that: $| cos\theta| \le \frac{1}{2}$, and hence: 
\begin{equation}\label{gaussEqu}
|\sin \theta| \ge \frac{\sqrt{3}}{2}.
\end{equation}
3. As shown in Theorem~\ref{lowerBoundTheorem},\\
\newline
 $det(u_1,u_2) \le S_h=\frac{\sqrt{3}}{2}h^2R^2 + 2 hR + (2+hR) \sqrt{2}$.\\
 \newline 
It results: $|u_1| |u_2| |\sin \theta| \le S_h.$ Using~\ref{gaussEqu}, we get:
\begin{equation}\label{v1Equa}
 \frac{\sqrt{3}}{2} |u_1||u_2| \le |u_1||u_2| | \sin \theta| \le S_h 
\end{equation}\\
And as $|u_1|\le |u_2|$, from~\ref{v1Equa} we have: 
\begin{equation}\label{v1SquareEqua}
\frac{\sqrt{3}}{2} |u_1|^2 \le S_h.
\end{equation}

We now separate the cases where $R > \sqrt{2}$ and $R \le \sqrt{2}$

\subsubsection{Case $R > \sqrt{2}$}

Using~\ref{minLengthEqu2}, from~\ref{v1SquareEqua} we get:
\begin{equation}
 \frac{\sqrt{3}}{2}|u_1|(R-\sqrt{2})h < \frac{\sqrt{3}}{2}|u_1|^2 \le S_h.
\end{equation}
And using~\ref{minLengthEqu2} in~\ref{v1Equa}: 
\begin{equation}
\frac{\sqrt{3}}{2}|u_2| (R-\sqrt{2})h < S_h.
\end{equation}

To summarize, the two generator vectors should verify for $R>\sqrt{2}$:
\begin{Equation}
{l_1}_{min} < |u_1| \le {l_1}_{max} \mathrm{~and~}
{l_2}_{min} < |u_2| < {l_2}_{max}
\end{Equation}
with: 
\begin{System}\label{vectModule}
{l_1}_{min} = h(R-\sqrt{2})  \\
{l_1}_{max} = \sqrt{\frac{2}{\sqrt{3}}S_h} \\ 
{l_2}_{min} = h(R-\sqrt{2})\\
{l_2}_{max} = \frac{2}{\sqrt{3}} \frac{S_h}{h(R-\sqrt{2})} 
\end{System}

In practice, to compute the two generator vectors, we determine
the upper and the lower bounds of the coordinates of $u_1$, and $u_2$ using the System~\ref{vectModule}.
Notice that we can search the vectors in the half plane $(y\ge0)$, because if $u_1$ and $u_2$ are generator vectors, then their symmetric vectors with respect to ($y=0$) axis are also generator vectors. Consequently, we have:

\begin{System}\label{v1coordinates}
{-l_1}_{max} \le x_1 \le {l_1}_{max} \\
0 \le y_1 \le {l_1}_{max}
\end{System}
and 
\begin{System}\label{v2coordinates}
-{l_2}_{max} \le x_2 \le {l_2}_{max} \\
0 \le y_2 \le {l_2}_{max}
\end{System}

\subsubsection{Case $R \le \sqrt{2}$}

We use ${l_1}_{min} = 0$. In addition, 
we can use the bound implied by the vectors proposed in 
Remark~\ref{rem:near-square-vector}, that is: $S_s = (hR+1)^2$, which replaces $S_h$
for the computation of ${l_1}_{max}$. Then, instead of a fixed bounds
for ${l_1}_{max}$ and ${l_2}_{max}$, we propose a bound depending on $u_1$, 
by using (\ref{v1Equa}) we have: ${l_2}_{max}(u_1) = \frac{2 \sqrt{3} S_s}{3 |u_1|} $, and because $|u_2|\ge|u_1|$ we have: ${l_1}_{min}(u_1) = |u_1|$. Hence:

\begin{System}\label{eq:set-u1-u2-small-R}
0 < |u_1| \le hR+1 \\
|u_1| \le |u_2| \le \frac{2 \sqrt{3} (hR+1)^2}{3 |u_1|}
\end{System}

\subsubsection{Method OVS}

To find the optimal vectors, we define Method OVS.\\
\mybox{
\begin{method}[OVS] $ $\\
1. The first step is to search $S_1$ the initial set of generator vectors $u_1$ and $u_2$. $S_1$ is the set of vectors having as coordinates the integers $(x_1,y_1)$ and $(x_2,y_2)$ verifying  System~\ref{v2coordinates} if $R > \sqrt{2}$
and System~\ref{eq:set-u1-u2-small-R} if $R \le \sqrt{2}$. \\
2. Now, the set $S_2$ should be filtered to keep only reduced and valid vectors. Indeed, for each couple of vectors $(u_1,u_2)$ in $S_1$, we should verify:\\
2.1. $(u_1,u_2)$ are reduced, that is they verify System~\ref{GaussPr}.\\
2.2. to check the validity of the coloring, two cases are possible: \\
2.2.1 if $R>\sqrt{2}$ apply Method VC2.\\
2.2.2. otherwise, apply Method VC1.\\
3. After the step 2, we  obtain the set of valid reduced vectors. Now, the optimal vectors are the vectors having the smallest absolute value of their determinant.\\

Notice that the search of the optimal vectors can be done by a central unit, that distributes the vectors to all nodes. It is also possible that each node in the grid computes the two generator vectors.
\end{method}
}

We can now evaluate the complexity of {\em VCM} that lies in the generator vectors computation and in their validity check. 
\begin{theorem}
Depending on {\em VC} method, {\em VCM} complexity is in $\Theta(R^6)$ for Method VC1 and $\Theta(R^4)$ for Method VC2.
\end{theorem}
\begin{proof}
The vector search phase is in $\Theta(R^2)$ for each vector. The validity check is in $\Theta(R^2)$ for Method VC1 and $\Theta(1)$ for Method VC2.
\end{proof}

\section{Summary: How to Apply VCM in Practice}\label{sec:vcm-practice}
This section summarizes the previous sections. In practice, to apply {\em VCM}, we start from a set of sensor nodes arranged as a two-dimensional
lattice (identified by their integer coordinates).
In reality, in an actual network, the set of neighbors will not be exactly 
given by the set of nodes within a fixed range $R$. However, notice that 
a valid $h$-hop coloring for a given $R$, is also a valid $h$-hop coloring for 
$R' < R$ (although likely non-optimal). Hence, we start by selecting the value of $R$, a radio range such as two nodes at a distance greater than $R$ are never neighbors
(may be using measurements or neighborhood detection), and a value $h$.
Then, each node proceeds as follows:\\
\noindent
1. Find the optimal valid vectors using the Method OVS.\\
2. Each node computes its color by applying either VCM-NCC~1 or VCM-NCC~2.\\

We can notice that {\em VCM} allows each node to know its color in a single round.

\section{Coloring Results with VCM}\label{resultsSec}

Note that further examples of colorings are available externally at 
\cite{bib:VCM-site}.

\subsection{Examples of Vectors}

Table~\ref{TableVectorGrid} gives for different radio ranges two vectors generating the optimal periodic pattern as well as the minimal number of colors obtained by a periodic pattern, for both a 2-hop coloring and a 3-hop coloring. The '*' symbol highlights the optimality of the number of colors used. 
\begin{table}[!h]
\caption{Vectors generating the optimal number of colors.}
\label{TableVectorGrid}
\begin{center}
\begin{tabular}{|c||c|c|c||c|c|c|}
\hline
Radio & \multicolumn{3}{|c||}{2-hop coloring}& \multicolumn{3}{|c|}{3-hop coloring}\\
           \cline{2-7}
range& vector1& vector2& colors& vector1& vector2& colors\\
\hline \hline
1& (2,1) & (-1,2)& 5*&(2,2) & (-2,2)& 8*\\
\hline
1.5& (-3,0)&(0,3)& 9*&(4,0) & (0,4)& 16*\\
\hline
2& (3,2)& (-2,3)& 13* &(4,3) & (-3,4)& 25*\\
\hline
2.5&(4,3)& (-1,5)& 23* &(5,5) & (-7,2)& 45*\\
\hline
3& (5,3)& (-1,6)& 33* &(7,5) & (-8,4)& 68*\\
\hline
3.5&(5,4)& (-6,3)& 39*& (8,5) & (-8,5)& 80*\\
\hline
4& (7,3)& (-6,5)& 53*&(8,8) & (-11,3)& 112*\\
\hline
4.5& (9,2) & (-6,7)& 75*&(13,3) & (-9,10)& 157*\\
\hline
5& (9,4)& (-1,10) & 94*&(14,4) & (3,15)& 198*\\
\hline
5.5& (9,6)& (-1,11)&105*&(16,0) & (8,14)& 224*\\
\hline
6&(11,4)& (-9,8)&124* & (17,4)& (-12,13)& 269*\\
\hline
6.5&(13,1)& (-7,11)&150*&(-19,0)&(9,17)& 323*\\
\hline
7& (10,9)&(-4,13)& 166*&(15,13) &(-19, 7)& 352*\\
\hline
\end{tabular}
\end{center}
\end{table}
\newline
We observe that 2-hop coloring of the grid with radio range $R=3$ is not equivalent to 3-hop coloring of a grid with $R=2$ in terms of the number f colors (33 vs. 25).
We conclude that the optimal number of colors is not determined only by the product $h*R$, but also the values of $h$ and $R$ separately.
\subsection{Comparison with Other Methods}
Table~\ref{TableVMVariousPrio} depicts the simulation results obtained with the {\em VCM} for various grids, with various radio ranges. The method is compared to a distributed coloring algorithm using line/column as priority assignment heuristics. As observed in Table~\ref{TableVariousPrioNodeNb}, the random priority assignment produces a high number of colors and hence is not presented here.
Results are given for 3-hop coloring. 
For high radio range values, the number of nodes should be high enough to allow the reproduction of the color pattern. 

\begin{table}[!h]
\caption{Number of colors obtained for 3-hop coloring.}
\label{TableVMVariousPrio}
\begin{center}
\begin{tabular}{|c|c|c|c|}
\hline
Radio & Grid size & \multicolumn{2}{|c|}{Colors} \\ 
\cline{3-4}
range & & VCM & line/column \\
\hline \hline
1 & 10x10 & 8* & 8*\\
\cline{2-4}
& 20x20 & 8* &8*\\
\cline{2-4}
& 30x30 & 8* &8*\\
\hline
\hline
1.5 & 10x10 & 16* & 16*\\
\cline{2-4}
& 20x20 & 16* &16*\\
\cline{2-4}
& 30x30 & 16* &16*\\
\hline
\hline
2 & 10x10 & 25* & 30\\
\cline{2-4}
& 20x20 & 25* & 33\\
\cline{2-4}
& 30x30 & 25* &33\\
\hline
\hline
3 & 20x20 & 68* & 80\\
\cline{2-4}
& 30x30 & 68* &83\\
\hline
\hline
3.5 & 20x20 & 80* & 91\\
\cline{2-4}
& 30x30 & 80* &91\\
\hline
\end{tabular}
\end{center}
\end{table}
We observe that {\em VCM} provides an optimal 3-hop coloring, for any radio range. This is not true for any other priority assignment tested. Moreover, the number of colors does not depend on the grid size. 

\section{The Vector Method and Real Wireless Networks}\label{sec:Realwireless}
In this paper, {\em VCM} has been described for grid topology since this topology is used by real applications briefly presented in Section~\ref{contextSec}.
Notice however, that wireless communication may differ from what is expected by the theory that often uses simplified models: radio links may be asymmetric, a radio link may exist even if the remote node is at a distance higher than the transmission range or conversely not exist even if the remote node is in the theoretical radio range. 
That is why, the first step in {\em VCM} is to select $R$ such that two nodes that are at a distance greater than $R$ are not neighbors.\\
Another real aspect in wireless sensor networks is nodes late arrival (because of the mobility or in case of late start-up) and nodes disappearance (a node is out of battery for instance). What is the impact of such impairments on {\em VCM}? We can classify these impairments in two categories:\\
1. Radio links disappearance: in this case, {\em VCM} always provides a valid coloring. The periodic coloring may still be optimal, as long as the percentage of missing radio links is below a given threshold $L_1$.\\
2. Radio links appearance: in this case, nodes that should not be neighbors (or heard nodes in case of asymmetric links) are. As a consequence, nodes having the same color may interfere because of these additional radio links. The periodic coloring provided by {\em VCM} may still be perfectly acceptable by the application as long as the percentage of additional radio links is below a given threshold $L_2$.\\

As a further work, we will evaluate the thresholds $L_1$ and $L_2$ and also study how random topologies can be mapped in grid topologies.

\section{Conclusion}\label{conclusion}
In this research report, we have presented a new method called {\em VCM}, the Vector-Based Coloring Method, able to provide an optimal periodic $h$-hop coloring of any grid, with $h$ an integer $\geq1$, for any radio range $R$. This method is easy to use: a single round is needed. It suffices to compute the two generator vectors, as shown in this paper. Knowing its coordinates within the grid, each node deduces its color from a simple computation given in the paper. We have shown that this $h$-hop node coloring is optimal in terms of colors and rounds. We determined also an upper and a lower bound for the number of colors needed to color an infinite grid. {\em VCM} provides the optimal number of colors compared to all possible coloring including non periodic ones. Finally, we discussed how to apply {\em VCM} in real wireless networks. 

\newpage
\section{Annex}
In this Annex, we group mathematical results and grid properties that are useful to study the validity and the optimality of {\em VCM}.

\subsection{Gauss Lattice Reduction}\label{sec:gauss-reduction}
For any pair of vectors $u_1$, $u_2$ generating 
a lattice $\Lambda(u_1,u_2)$, the Gauss lattice reduction algorithm
provides two {\em reduced} vectors
$v_1,v_2$ generating exactly the same lattice and 
verifying the System~(\ref{GaussPr})
\begin{System}\label{GaussPr}
|v_1| \le |v_2| \\
2 | v_1 \cdot v_2 | \ \le\ |v_1|^2
\end{System}
Additional properties are that $v_1$ and $v_2$ are also the two shortest
distinct vectors of  $\Lambda(u_1,u_2)$ 
and have the same lattice determinant as $u_1,u_2$.
See for instance \cite{VV07} for more details.
\newline

\subsection{Relation between Number of Hops and Actual Distance}
Hereafter, we introduce some results related to grid networks. These results can be applied to {\em VCM}, or any other algorithm.

Results in this section are inequalities, establishing links between 
number of hops and actual distance.

\begin{lemma}
\label{distmaxg}
For any point $V$ of $\mathbb{R}^2$, there exists a node $V'$ of the grid $\ZZ^2$ such that $d(V,V') \leq \sqrt{2}/2$.
\end{lemma}
\proof
In the worst case, the node $V$ occupies the center of a grid cell. It is at equal distance of two grid nodes that are diagonally opposed. Hence, its distance to one of them is equal to $\sqrt{2}/2$.
\endproof
\begin{lemma}
\label{hhopprop}
For any transmission range $R > \sqrt{2}$, for any grid node $U$, any node $V$ that meets $d(U,V) \leq (R-\sqrt{2})h$ is at most $h$-hop away from $U$.
\end{lemma}
\proof
We consider the $h-1$ points of $\mathbb{R}^2$ that allow us to divide the segment $[U,V]$ in $h$ equal parts.\\ Let $W_i$ be these nodes, with $i \in [1,h-1]$.\\ For any $i \in [1,h-1]$, let $W'_i$ the grid point the closest to $W_i$. For simplicity reason, we denote $W'_0=U$ and $W'_h=V$. We have $d(U,V) \leq \sum_{i=0}^{h-1} d(W'_i,W'_{i+1})$.\\ 
We have $d(W'_i,W'_{i+1}) \leq d(W'_i,W_i) + d(W_i,W_{i+1}) + d(W_{i+1},W'_{i+1})$.
According to Lemma~\ref{distmaxg}, we have $d(W_i,W'_i) \leq \sqrt{2}/2$. Hence, we get $d(W'_i,W'_{i+1}) \leq \sqrt{2}+ d(W_i,W_{i+1})$. By construction, $d(W_i,W_{i+1})=d(U,V)/h$. \\
If $d(U,V) \leq (R-\sqrt{2})h$, then $d(W'_i,W'_{i+1})\leq R$. Hence, nodes $W_i$ for $i \in \{1,2, \ldots, h-1 \}$ constitute a h-hop path between $U$ and $V$.
\endproof 
\begin{lemma}\label{morethanHhop}
For any transmission range $R$, for any two grid nodes $U$ and $V$, in $h$-hop coloring, if $d(U,V)>hR$ then $U$ and $V$ are at least $(h+1)$-hop away. 
\end{lemma}
\proof
By contradiction assume that, $U$ and $V$ are $h$-hop away or less. Let $W_i$ be the $k-1$ nodes constituting the $k$-hop path between $U$ and $V$, with $k\le h$. Let $W_1=U$, and $W_h=V$. Since nodes $W_i$ are 1-hop neighbors, we have:
\begin{eqnarray*}
|\vec{UV}| = |\sum_{i=1}^{h}\vec{W_{i}W_{i+1}}|  \le  \sum_{i=1}^{h}d(W_{i},W_{i+1})
      \le  hR. 
\end{eqnarray*}
Hence the contradiction.
\endproof

~\newline
Let $U$, $V$ be two points of $\ZZ^2$ and 
define $\mathcal{H}(U,V)$ as the number of hops between $U$ and $V$ (it
is an integer). For any $R>0$ (some inequalities are trivially true
when $R \le \sqrt{2}$), the
lemma~\ref{hhopprop} and lemma~\ref{morethanHhop} can be summarized as:
\begin{eqnarray}
d(U,V) \le (R-\sqrt{2})h & \implies & \mathcal{H}(U,V) \le h \\
d(U,V) > (h-1)R & \implies & \mathcal{H}(U,V) \ge h \\
\mathcal{H}(U,V) \ge h & \implies & d(U,V) > (R-\sqrt{2})(h-1) \\
\mathcal{H}(U,V) \le h & \implies & d(U,V) \le hR
\end{eqnarray}

\subsection{Bounds on Distance and Number of Hops of Points on a Lattice}

\begin{lemma}\label{fromGauss}
If $u_1$ and $u_2$ are reduced generator vectors of a lattice $\Lambda(u_1,u_2)$, with $|u_1|\le|u_2|$, then for any vector $w$ such that $w=\alpha u_1+\beta u_2$, and $\alpha, \beta \in \ZZ^2$, we have $|w| \ge\frac{3}{4}\alpha^2 |u_1|^2$, and $|w|\ge\frac{3}{4}\beta^2 |u_1|^2$.
\end{lemma}
\proof
Let $W$ the node of coordinates $(\alpha, \beta)$. We have:
\begin{small}
\begin{eqnarray*}
|\vec{UW}|^2 & = & \alpha^2 |\myvec{1}|^2+ \beta^2 |\myvec{2}|^2+2\alpha\beta \ (\myvec{1} \cdot \myvec{2})\\	 
& \ge & \alpha^2 |\myvec{1}|^2+ \beta^2 |\myvec{2}|^2 - 2| \alpha | | \beta| \ |\myvec{1} \cdot \myvec{2}|.
\end{eqnarray*}
\end{small}
Since $u_1$ and $u_2$ are reduced vectors, we can use the property given in the System~\refeq{GaussPr}, we get:
\begin{eqnarray*}
|\vec{UW}|^2 & \ge & \alpha^2 |\myvec{1}|^2+ \beta^2 |\myvec{1}|^2 +| \alpha | | \beta | |\myvec{1}|^2. \mathrm{~Hence,~} \\
|\vec{UW}|^2 & \ge & ((|\alpha| - |\beta|)^2 +| \alpha | | \beta |) |\myvec{1}|^2.
\end{eqnarray*}
Notice that this quantity does not change if we change the sign of $\alpha$ or of $\beta$. Thus, we assume $(\alpha\ge0)$, $(\beta\ge0)$, and let $f(\alpha,\beta) = (\alpha - \beta)^2 +\alpha\beta$. \\
By a change of variable $\beta=\frac{\alpha}{2}+\lambda$, we get:\\ 
$f(\alpha,\beta)=\frac{3}{4}\alpha^2+\lambda^2 \ge \frac{3}{4}\alpha^2$. Similarly, we have $f(\alpha,\beta)\ge \frac{3}{4}\beta^2 $. Hence the lemma.
\endproof

\begin{lemma}
\label{lem:mu-bound}
Consider any transmission range $R>\sqrt{2}$, two reduced generator vectors $u_1$ and $u_2$ of the lattice $\Lambda(u_1,u_2)$, and a node $W$ 
with $\vec{UW} = \alpha u_1 + \beta u_2$
for some $\alpha$ and $\beta$ in $\ZZ$. 
Assume also that the point $V_1$ such that $\vec{UV_1} = u_1$ 
is at strictly more than $h$ hops from the $U$.
Then:\newline
if $|\alpha| \ge \mu(R) $ or $|\beta| \ge \mu(R)$ where $\mu(R) = \frac{2\sqrt{3}R}{3(R-\sqrt{2})}$, we have: $W$ is strictly more than $h$ hops away from $U$.
\end{lemma}
\proof
Since $u_1$ and $u_2$ are reduced, we can apply Lemma~\ref{fromGauss} and obtain:
\begin{eqnarray*}
|\vec{UW}|^2 & \ge & f(\alpha,\beta)|\myvec{1}|^2\\
				& \ge & \frac{3}{4}\alpha^2 |\myvec{1}|^2, \mathrm{~and~as~} \alpha \ge\mu(R) \\
				& \ge &\frac{3}{4} (\frac{2\sqrt{3}R}{3(R-\sqrt{2})})^2 |\myvec{1}|^2 \\
& \ge &(\frac{R}{R-\sqrt{2}})^2 |\myvec{1}|^2
\end{eqnarray*}  
Since the point $V_1$ is strictly more than h-hop away from $U$,
the lemma~\ref{hhopprop} implies by contradiction that 
$|u_1| = |\vec{UV_1}| > (R - \sqrt{2})h$. It follows that:
$$|\vec{UW}|^2 > R^2h^2$$

Applying Lemma~\ref{morethanHhop}, we obtain the result.
\endproof

\end{document}